\documentclass[11pt]{article}
\usepackage[margin=1in]{geometry}
\date{\vspace{-5ex}}

\usepackage{graphicx}
\usepackage{caption}
\usepackage{float}
\usepackage{xcolor}
\usepackage{graphics}
\usepackage{cite}
\usepackage{caption}
\usepackage{epstopdf}
\usepackage{amsmath}
	\usepackage{amssymb}
	 \usepackage{enumerate}
	 \usepackage{balance}
	 \usepackage{empheq}
	 \usepackage{mathtools}
	 \usepackage{algorithm}
\usepackage{algpseudocode}
\usepackage{tabularx}

\newcommand{\tr}{{{\mathsf T}}}
\DeclareMathOperator*{\minimize}{minimize}
\DeclareMathOperator*{\st}{subject~to}

\usepackage{hyperref}
\hypersetup{
    colorlinks=true,
    linkcolor=black,
    filecolor=black,
    urlcolor=black,
}

\mathchardef\Re="023C
\mathchardef\Im="023D

\newtheorem{theorem}{Theorem}
\newtheorem{lemma}{Lemma}
\newtheorem{proposition}{Proposition}

\newtheorem{definition}{Definition}
\newtheorem{proof}{Proof}

\newtheorem{example}{Example}

\begin{document}

\title{  \LARGE \bf  On Separable Quadratic Lyapunov Functions for Convex Design of Distributed Controllers 
}

	 \author{Luca Furieri\thanks{This research was gratefully funded by the European Union ERC Starting Grant CONENE. Luca Furieri and Maryam Kamgarpour are with the Automatic Control Laboratory, Department of Information Technology and Electrical Engineering, ETH Z\"{u}rich, Switzerland. E-mails: {\tt\footnotesize \{furieril, mkamgar\}@control.ee.ethz.ch}} \and Yang Zheng\thanks{Yang Zheng and Antonis Papachristodoulou are with the Department of Engineering Science, University of Oxford, United Kingdom. E-mails: {\tt\footnotesize \{yang.zheng, antonis\}@eng.ox.ac.uk}} \and Antonis Papachristodoulou$^\dagger$ \and Maryam Kamgarpour$^\ast$
	%
}
\maketitle

\begin{abstract}                          
We consider the problem of designing a stabilizing and optimal static controller with a pre-specified sparsity pattern. Since this problem is NP-hard in general, it is necessary to resort to approximation approaches. In this paper, we characterize a class of convex restrictions of this problem that are based on designing a separable quadratic Lyapunov function for the closed-loop system. This approach generalizes previous results based on optimizing over diagonal Lyapunov functions, thus allowing for improved feasibility and performance. Moreover, we suggest a simple procedure to compute favourable structures for the Lyapunov function yielding high-performance distributed controllers. Numerical examples validate our results.
\end{abstract}

\section{Introduction}
 Recent years have witnessed a growth in the sensing and actuation capabilities of control systems. These technological advances  enable addressing a wide range of  engineering applications, such as the smart grid \cite{dorfler2014sparsity}, biological networks \cite{prescott2014layered}, and automated highways \cite{zheng2017distributed}. These applications commonly rely on efficiently coordinating
the decision making of multiple interacting agents, which only have partial information about the internal variables of the overall system. The lack of full information often presents itself as structural constraints on the controllers' parameters and motivates the field of distributed control.

 It is well-known that synthesizing optimal controllers under structural constraints is a challenging task and amounts to an NP-hard problem in general \cite{tsitsiklis1984complexity,Witsenhausen}. A line of work has focused on  identifying particular interactions between the structural constraints and the system dynamics for which dynamic linear controllers are optimal and can be found in a tractable way \cite{ho1972team,bamieh2005convex}. These concepts were generalized under the notion of quadratic invariance (QI)  by using the Youla-Kucera parametrization \cite{rotkowitz2006characterization}. When QI does not hold,  convex approximations of the otherwise intractable problem were introduced in \cite{rotkowitz2012nearest} and generalized in \cite{furieri2018robust}.
  Computing optimal dynamic controllers is challenging in general, as it involves approximating infinite dimensional programs\cite{Alavian}. Hence, significant work has focused on the synthesis of static feedback controllers. Thanks to their simplicity, structured static controllers are also commonly employed in the field of distributed model predictive control (DMPC) in order to generate robustly invariant control policies with partial information \cite{conte2012distributed,darivianakis2018high}. However, computing optimal and stabilizing static controllers is also intractable with general structural constraints. 

The work in \cite{lin2013design} has considered a nonlinear optimization technique to find locally optimal distributed static controllers. Rather than directly tackling this intractable problem, the authors in \cite{SDP} have developed convex relaxations based on dropping rank constraints and provided optimality bounds in terms of low-rank solutions. However, these relaxations might fail to recover a distributed controller that is stabilizing. Polynomial optimization has been used  in \cite{wang2018convex} to obtain a sequence of convex relaxations which converges to a stabilizing distributed controller. Nevertheless, performance of the recovered solution is not directly addressed in \cite{wang2018convex}. A convex surrogate based on approximating the non-convex cost function with matrix norms was proposed in \cite{dvijotham2015convex}. However, the suboptimality bounds can be loose.

A different approach is to consider a \emph{convex restriction}, where the unstructured problem is reformulated with an equivalent convex program involving additional convex constraints to guarantee the structure of controllers. The advantage of a convex restriction is that its optimal solution can be readily computed by using standard convex optimization techniques. Furthermore, all its feasible solutions are structured and stabilizing by design. A disadvantage is that a restriction may be infeasible even when the original problem is feasible. Following this approach,  \cite{geromel1994decentralized, rubio2013static} proposed preserving the sparsity of the distributed controller by imposing a diagonal structure on a matrix variable. When the overall system is divided into interconnected subsystems with local information, the works \cite{conte2012distributed,zheng2017convex, han2017hierarchical} suggested forcing the matrix defining a quadratic Lyapunov function for the closed-loop system to be block-diagonal, where the blocks have predetermined dimensions. The authors in \cite{polyak2013lmi} proposed a different technique that is only applicable when the desired controller structure is either row or column sparse. In this paper, we improve upon these conservative convex restriction approaches up to their theoretical boundaries and we extend their applicability to general systems and information structures.

Our  contributions are as follows. First, in Section~\ref{se:allfeasibleconvex} we introduce the notion of \emph{sparsity invariance} to characterize a novel class of convex restrictions that is based on imposing appropriate sparsity patterns on certain matrix factors. As a result, we generalize previous approaches \cite{conte2012distributed, han2017hierarchical,geromel1994decentralized, zheng2017convex, rubio2013static} and we improve their feasibility and performance. Moreover, we show that convex restriction approaches based on sparsity invariance cannot be generalized further.  Second, we provide necessary and sufficient conditions for the feasibility of these convex restrictions, in terms of the existence  of a corresponding separable Lyapunov function for the closed-loop system (Section~\ref{se:allfeasibleconvex}). Third, we suggest a computationally tractable procedure to design favourable structures for the Lyapunov matrix to achieve high performance in Section~\ref{se:optimized}. This procedure highlights that increasing the performance is closely linked to loosening the degree of separability we force on the Lyapunov function. We validate our results through numerical examples both in Section~\ref{se:allfeasibleconvex} and Section~\ref{se:numerics}.

\section{Preliminaries}
\label{se:preliminaries}
In this section, we first introduce some notation on sparsity structures, and then present the problem statement of  distributed optimal control. We highlight its non-tractability and introduce the class of convex restrictions under investigation.
\subsection{Notation and sparsity structures}	
 We use $\mathbb{R}$, $\mathbb{C}$ and $\mathbb{N}$ to  denote the sets of real numbers, complex numbers and  positive integers, respectively. For any $n \in \mathbb{N}$, we use $\mathbb{N}_{[1,n]}\subseteq \mathbb{N}$ to denote the set of integers from $1$ to $n$.  The $(i,j)$-th element in a matrix $Y \in \mathbb{R}^{m \times n}$ is referred to as $Y_{ij}$. We use $I_n$ to denote the identity matrix of size $n \times n$,  $0_{m \times n}$ to denote the zero matrix of size $m \times n$  and $1_{m\times n}$ to denote the matrix of all ones of size $m \times n$. The vector $e_i \in \mathbb{R}^n$ denotes the $i$-th vector of the standard basis of $\mathbb{R}^n$, having a $1$ in its $i$-th position and $0$ everywhere else. For a square matrix $M \in \mathbb{R}^{n \times n}$ we write $M=\text{blkdiag}(M_1,\cdots, M_r) \in \mathbb{R}^{n \times n}$ if $M$ is block-diagonal with $M_i \in \mathbb{R}^{n_i \times n_i}$ on its diagonal entries, and we have $\sum_{i=1}^r n_i=n$.  The notation $M\succ 0$ (resp. $M \succeq 0$) refers to $M$ being symmetric and positive definite (resp. symmetric and positive semidefinite).

The sparsity structure of a matrix can be conveniently represented by a binary matrix. A binary matrix is a matrix with entries from the set $\{0,1\}$, and we use $\{0,1\}^{m \times n}$ to denote the set of $m \times n$ binary matrices.
Given a binary matrix $X \in \{0,1\}^{m \times n}$, we define the sparsity subspace $\text{Sparse}(X) \subseteq \mathbb{R}^{m \times n}$ as
\begin{equation*}
\text{Sparse}(X) := \{Y \in \mathbb{R}^{m \times n}\mid Y_{ij} = 0 \; \text{if}\; X_{ij}=0, \forall i,j\}.
\end{equation*}
Similarly, given $Y\in \mathbb{R}^{m \times n}$, we define a binary matrix $X\, := \, \text{Struct}(Y)$ encoding the sparsity pattern of $Y$ as
$$
    X_{ij}=\begin{cases}
	   1 & \text{if}\;Y_{ij} \neq 0\,,\\
	   0 & \text{otherwise}\,.
	\end{cases}
$$

 Let $X, \hat{X} \in \{0,1\}^{m \times n}$ and $Z \in \{0,1\}^{n \times p}$ 
be binary matrices. Throughout the paper, we adopt the following conventions:
    \begin{equation*}
        X + \hat{X} := \text{Struct}(X + \hat{X}), ~~
        XZ:=\text{Struct}(XZ)\,.
    \end{equation*}
We state that $X \leq \hat{X}$ if and only if $X_{ij}\leq \hat{X}_{ij}\;\forall i,j$, and $X < \hat{X}$ if and only if $X \leq \hat{X}$ and there exist indices $i,j$ such that $X_{ij}<\hat{X}_{ij}$. {Also, we denote} $X \nleq \hat{X}$ if and only if there exist indices $i,j$ such that $X_{ij}>\hat{X}_{ij}$.
%

A permutation matrix $\Pi \in \{0,1\}^{n \times n}$ is a binary matrix that has exactly one entry with $1$ in each row and each column and $0$ everywhere else.

An undirected graph $\mathcal{G}=(\mathcal{V},\mathcal{E})$ is defined by a set of nodes $\mathcal{V}$ and a set of edges $\mathcal{E} \subseteq \mathcal{V} \times \mathcal{V}$, where $(i,j) \in \mathcal{E} \Leftrightarrow (j,i) \in \mathcal{E}$. Given a symmetric binary matrix $X \in \{0,1\}^{n \times n}$, we denote the undirected graph having $X$ as its adjacency matrix as $\mathcal{G}(X)$. Then, any binary matrix $X=X^\mathsf{T} \in \{0,1\}^{n \times n}$ corresponds to an undirected graph $\mathcal{G}(X)$, and vice-versa.  The transitive closure of a graph $\mathcal{G}=(\mathcal{V},\mathcal{E})$  is defined as the graph $\tilde{\mathcal{G}}=(\mathcal{V},\tilde{\mathcal{E}})$ where there is an edge between $\nu_i \in \mathcal{V}$ and $\nu_j \in \mathcal{V}$ in $\tilde{\mathcal{E}}$ if and only if there is a  path in $\mathcal{G}$ between $\nu_i$ and $\nu_j$. A connected component of a graph $\mathcal{G}(\mathcal{V},\mathcal{E})$ is a subgraph $\mathcal{G}'(\mathcal{V}',\mathcal{E}')$ with $\mathcal{V}'\subseteq \mathcal{V}$ and $\mathcal{E}' \subseteq \mathcal{E}$ in which any two vertices of $\mathcal{G}'(\mathcal{V}',\mathcal{E}')$ are connected to each other by a path, and  such that for all $\nu' \in \mathcal{V}'$ and $\nu \in \mathcal{V}\setminus \mathcal{V}'$  there is no edge in $\mathcal{E}$  connecting them. For any symmetric binary matrix $X\geq I_n$, $\mathcal{G}(X^{n-1})$ is the transitive closure of $\mathcal{G}(X)$. Since $\mathcal{G}(X)$ is undirected, its transitive closure $\mathcal{G}(X^{n-1})$ is a graph that consists of complete subgraphs (also denoted as cliques) corresponding to the connected components of $\mathcal{G}(X)$ \cite{biggs1993algebraic}.

\subsection{Problem statement}
We consider a linear dynamical system
\begin{equation}\label{eq:DynamicsGlobal}
    \dot{x}(t) = Ax(t) +Bu(t) + Hw(t),
\end{equation}
where $x(t) \in \mathbb{R}^{n}$, $u(t)\in \mathbb{R}^{m}$ and $w(t) \in \mathbb{R}^{q}$ denote the state, control input, and disturbance vectors at time $t$, respectively. We look for a linear static state feedback controller
\begin{equation}\label{eq:ControllerGlobal}
    u(t) = Kx(t),  \quad K \in \mathcal{S},
\end{equation}
where $\mathcal{S}:=\text{Sparse}(S)$ denotes the  subspace of matrices $K$ having the sparsity pattern specified by $S \in \{0,1\}^{m \times n}$. Sparsity requirements on $K$ are common in distributed control.  Indeed, by choosing the binary matrix $S$ such that $S_{ij}=0$, we can encode the requirement that the $i$-th  control input cannot be a function of the $j$-th  state variable.  In other words, by choosing the subspace $\mathcal{S}$ appropriately, we can encode information constraints in our control problem.   The closed-loop system is then
\begin{equation}\label{eq:Closedloop}
    \dot{x}(t) = (A+BK)x(t)+Hw(t), \quad K \in \mathcal{S}.
\end{equation}
In this paper, we address the problem of computing the stabilizing static linear feedback policy $u(t)=Kx(t)$ which minimizes a specified norm  of the closed-loop transfer function from disturbances $w(t)$ to a performance signal defined as follows
\begin{equation}
\label{eq:performance_signal}
z(t)=Cx(t)+Du(t)\,,
\end{equation}
where $C \in \mathbb{R}^{p \times n}$ and $D \in \mathbb{R}^{p \times m}$. The corresponding distributed control problem can be written as follows:
	\begin{empheq}[box=\fbox]{alignat=3}
	&\textbf{Problem }&&\mathcal{P}_K \nonumber\\
	 &\minimize_{K } ~~&& ||(C+DK)(sI-A-BK)^{-1}H|| \label{eq:H2norm}\\
	& \st &&(A+BK) \text{ is Hurwitz}\,, \quad  K \in \mathcal{S}\,, \nonumber
	\end{empheq}	
	where $s\in \mathbb{C}$  and usual choices for the norm $||\cdot||$ are the $\mathcal{H}_2$ and the $\mathcal{H}_\infty$ functionals. 	The constraint that $(A+BK)$ is Hurwitz ensures that the cost is finite. Necessary conditions for feasibility of problem $\mathcal{P}_K$ are that the pair $(A,B)$ is stabilizable and that there are no  distributed fixed modes with respect to $\mathcal{S}$ \cite{alavian2014stabilizing}. Sufficient conditions for distributed stabilizability using static feedback are not known for general systems. For simplicity, we will only consider continuous-time systems with the goal of minimizing the $\mathcal{H}_2$ norm. However, we remark that the results of this paper can be easily extended to discrete-time systems and for the $\mathcal{H}_\infty$ norm.
	
		 It is immediate to verify that the optimization problem $\mathcal{P}_K$ is non-convex in $K$ in its present form, as the cost function and the requirement that $(A+BK)$ is Hurwitz are non-convex in general. Additional effort is thus needed to make this problem tractable and solvable with standard optimization techniques.

	Similar to \cite[Chapter 10]{boyd1994linear}, but with the addition of the structural constraint $\mathcal{S}$, problem $\mathcal{P}_K$ can be written as follows:
\begin{alignat}{3}
\medmuskip=-2mu
\thinmuskip=-2mu
\thickmuskip=-2mu
\nulldelimiterspace=-1pt
\scriptspace=0pt
&\minimize_{X,Y,Z}\hspace{0.18cm}&&~\text{Trace}\left(CXC^\mathsf{T}+DYC^\mathsf{T}+CY^\mathsf{T}D^\mathsf{T}+DZD^\mathsf{T}\right)\nonumber \\
&	\st &&~\begin{bmatrix}
Z&Y\\Y^\mathsf{T}&X
\end{bmatrix} \succeq  0 ,\quad X \succ 0\,, \label{eq:Z} \\
&~&& \text{$AX$\hspace{-0.01cm}$+$$XA^\mathsf{T}$\hspace{-0.01cm}$+$$BY$\hspace{-0.01cm}$+$$Y^\mathsf{T}B^\mathsf{T}$\hspace{-0.01cm}$+$$HH^\mathsf{T}\prec~ 0$}\label{eq:Lyap_X}\,,\\
&~&& YX^{-1} \in \mathcal{S}\,.\nonumber
\end{alignat}
By pre-and post-multiplying (\ref{eq:Lyap_X}) by $P=X^{-1}$, we derive that $V(x)=x^\mathsf{T}Px$ is a Lyapunov function for the closed-loop system. The solution to the original problem $\mathcal{P}_K$ is recovered as $K=YX^{-1}$. Without the structural constraint $\mathcal{S}$, the above reformulation will be a convex semidefinite program solvable with efficient optimization techniques. The primary source of non-convexity is the nonlinear constraint $YX^{-1} \in \mathcal{S}$.

\subsection{A class of convex restrictions}
Our underlying idea is to simplify the nonlinear constraint $YX^{-1} \in \mathcal{S}$ by requiring instead that $Y$ and $X$ have certain distinct sparsity patterns. In other words, we look for distributed controllers $K \in \mathcal{S}$ by restricting our search over structured matrix factors. We will thus study the following convex optimization problem:

	\begin{empheq}[box=\fbox]{alignat*=3}
	\medmuskip=2mu
\thinmuskip=2mu
\thickmuskip=2mu
	&\textbf{Problem }&&\mathcal{P}_{T,R}:~~ \nonumber\\
	&\minimize_{X,Y,Z} &&~\text{Tr}\left(CXC^\mathsf{T}+DYC^\mathsf{T}+CY^\mathsf{T}D^\mathsf{T}+DZD^\mathsf{T}\right)\nonumber\\
 &\st && ~(\ref{eq:Z}),~  (\ref{eq:Lyap_X}),~ Y \in \text{Sparse}(T), ~ X \in \text{Sparse}(R)\,,
	\end{empheq}
where $T \in \{0,1\}^{m \times n}, R \in \{0,1\}^{n \times n}$ are binary matrices to be designed. For the rest of the paper, we will assume that $R$ is symmetric with $R\geq I_n$ without loss of generality. Indeed, it is required in (\ref{eq:Z}) that $X\succ 0$. This implies  that $e_i^\mathsf{T}Xe_i=X_{ii}>0$ for each $i$. Therefore, the structure of $X$ must be symmetric and the entries of the diagonal of $X$ must be strictly positive.

Problem $\mathcal{P}_{T,R}$ is convex by construction and solvable via existing solvers. Our main goal is then to answer two fundamental questions:
\begin{enumerate}
\item under which conditions do the solutions of $\mathcal{P}_{T,R}$ recover feedback controllers $K = YX^{-1} \in \mathcal{S}$?
\item under which conditions is $\mathcal{P}_{T,R}$ feasible?
 \end{enumerate}


\section{Feasible Convex Restrictions Based on Separable Lyapunov Functions}
\label{se:allfeasibleconvex}
In this section we address the two questions raised above, by providing conditions for $\mathcal{P}_{T,R}$ to be a feasible restriction of $\mathcal{P}_K$.
\subsection{Sparsity invariance for convex restrictions}
\label{sub:sparsity}

Our approach is to characterize the set of  binary matrices $T$ and $R=R^\mathsf{T} \geq I_n$ such that  
\begin{align}
	\medmuskip=2mu
\thinmuskip=2mu
\thickmuskip=2mu
\label{eq:sparsity_invariance}
Y \in \text{Sparse}(T) \text{ and }X \in \text{Sparse}(R)\, \Rightarrow YX^{-1} \in \mathcal{S} .
\end{align}
 We  refer to the property (\ref{eq:sparsity_invariance}) as \emph{sparsity invariance}.  In order to address question $1)$ stated at the end of Section~\ref{se:preliminaries}, we give a full characterization of  sparsity invariance in the following theorem. Its proof is reported in the Appendix.


\begin{theorem}
\label{th:sparsity_invariance}
Let $T \in \{0,1\}^{m \times n}$ and $R \in \{0,1\}^{n \times n}$ be symmetric with $R\geq I_n$. Consider the following statements.
\begin{enumerate}
\item Sparsity invariance as per (\ref{eq:sparsity_invariance}) holds.
\item $T\leq S$ and $TR^{n-1} \leq S$.
\item Problem $\mathcal{P}_{T,R^{n-1}}$ is a convex restriction of $\mathcal{P}_K$.
 \end{enumerate}
 Then $1) \Leftrightarrow 2) \Rightarrow 3)$. 
\end{theorem}

We remark that equivalence of $1)$ and $2)$ of Theorem~\ref{th:sparsity_invariance} provides a complete characterization of all admissible sparsities for the matrix factors $Y$ and $X$  to ensure $K \in \mathcal{S}$, whereas the previous works \cite{geromel1994decentralized, zheng2017convex, rubio2013static, conte2012distributed, han2017hierarchical}  only considered a trivial case where $X$ is (block-)diagonal and $Y \in \mathcal{S}$.


In addition, for each $T$ and $R$ as per (\ref{eq:sparsity_invariance}), it is always preferable to  solve the convex restriction $\mathcal{P}_{T,R^{n-1}}$ instead of $\mathcal{P}_{T,R}$. Indeed, notice that if $TR^{n-1} \leq S$, then $T(R^{n-1})^{n-1} \leq S$ as a consequence of Cayley-Hamilton's theorem and the fact that $R\geq I_n$. Equivalently, when $T$ and $R$ satisfy sparsity invariance (\ref{eq:sparsity_invariance}), so do $T$ and $R^{n-1}$, and both $\mathcal{P}_{T,R}$ and $\mathcal{P}_{T,R^{n-1}}$ are convex restrictions of $\mathcal{P}_K$. Since requiring  $X \in \text{Sparse}(R')$ for some $R' < R^{n-1}$ is never convenient in terms of performance due to $\text{Sparse}(R') \subset \text{Sparse}(R^{n-1})$, we will mainly focus on the convex restriction $\mathcal{P}_{T,R^{n-1}}$ for the rest of the paper.


We proceed with addressing question $2)$ stated at the end of Section~\ref{se:preliminaries} about the feasibility of $\mathcal{P}_{T,R^{n-1}}$. It turns out that the feasibility of $\mathcal{P}_{T,R^{n-1}}$  is closely related to the existence of a quadratic Lyapunov function for the closed-loop system, which is \emph{separable} in the sense defined below.

\begin{definition}[Separable Lyapunov functions]
   \emph{Consider a linear system $\dot{x}(t) = A x(t)$, where $x(t) \in \mathbb{R}^{n}$. A quadratic function $V(x) = x^{\tr}Px$ that satisfies $P \succ 0, A^\tr P + P A \prec 0$ is a Lyapunov function for  the system. A quadratic Lyapunov function is separable if there exists a permutation matrix $\Pi$ such that
    \begin{equation} \label{eq:separableLy}
        \Pi P \Pi^\mathsf{T}=\text{blkdiag}\left( {P}_1, \cdots,~{P}_r\right),
    \end{equation}
    where $r\geq 1$, $v_i \in \mathbb{N}$, $P_i\succ 0$ has dimensions $v_i \times v_i$ for every $i$ and $\sum_{i=1}^r v_i = n$. More precisely, upon denoting the permuted state $\Pi x$ as
$
    \begin{bmatrix}z_1^\mathsf{T},\cdots,z_r^\mathsf{T}\end{bmatrix}^{\mathsf{T}}=\Pi x,
$
with $z_i \in \mathbb{R}^{v_i},$ $i\in \mathbb{N}_{[1,r]}$, we have
$
V(x) = \sum_{i=1}^r z_i^{\tr} P_i z_i.
$}
\end{definition}

 An extreme case is $v_i = 1$ for all $i$ and $r = n$, where we have a Lyapunov function $V(x) = x^{\tr}Px$ with a diagonal $P$ which is separable into $n$ addends. The other extreme case is $r=1$ and $v_1=n$, where $V(x)$ is not separable into multiple addends. An intermediate case is that of \cite{conte2012distributed,zheng2017convex, han2017hierarchical}, where $P$ is forced to be block-diagonal and the dimensions of each block matches that of a corresponding subsystem. Instead, our notion generalizes the cases above to all separable quadratic Lyapunov functions. Additionally, we do not restrict ourselves to the case of interconnected subsystems.
 
%



\subsection{Separable Lyapunov functions for feasible convex restrictions}
Here, we characterize the feasibility of our convex restrictions of $\mathcal{P}_K$ in terms of Lyapunov theory.
\begin{theorem}
\label{th:feasible_restriction}
Let $T \in \{0,1\}^{m \times n}$ and $R \in \{0,1\}^{n \times n}$ be symmetric with  $R\geq I_n$. The following two statements are equivalent.
\begin{enumerate}
\item Problem $\mathcal{P}_{T,R^{n-1}}$ is feasible.
\item There exist  $Y \in \text{Sparse}(T)$  and $P \in \text{Sparse}(R^{n-1})$ with $K=YP$ satisfying
\begin{equation}
    \label{eq:Lyap_P}
(A+BK)^\mathsf{T}P+P(A+BK)+PHH^\mathsf{T}P \prec 0\,.
\end{equation}
The function $V(x)=x^{\mathsf{T}}Px$ is a Lyapunov function for the closed-loop system (\ref{eq:Closedloop}) which is separable into $r$ addends, where $r$  is the number of connected components of $\mathcal{G}(R)$.
\end{enumerate}
\end{theorem}

The proof of Theorem~\ref{th:feasible_restriction} relies on the following lemma.
 \begin{lemma}
 \label{le:symmetry_necessity}
Given a symmetric $R \in \{0,1\}^{n \times n}$ with $R\geq I_n$, there exists a permutation matrix $\Pi \in \{0,1\}^{n \times n}$ such that 
 $\Pi R^{n-1}\Pi^\mathsf{T}=\text{blkdiag}\left( 1_{v_1 \times v_1}, ~1_{v_2 \times v_2},\cdots, ~1_{v_r \times v_r}\right)\,,$
 where $r$ 
 is the number of connected components of graph $\mathcal{G}(R)$ and $v_i$ is the number of nodes in 
 the $i$-th connected component for $i \in \mathbb{N}_{[1,r]}$.
 \end{lemma}
 \begin{proof}
 \emph{Let $\mathcal{G}(R)$ be the graph having $R=R^\mathsf{T} \geq I_n$ as its adjacency matrix. It is well known that for every graph $\mathcal{G}(R')$ isomorphic to $\mathcal{G}(R)$ a permutation matrix $\Pi$ such that $R'=\Pi R\Pi^\mathsf{T}$ exists  \cite{biggs1993algebraic}. Since $\mathcal{G}(R^{n-1})$ is isomorphic to a graph consisting of separate complete subgraphs,  then a permutation matrix such that $\Pi R^{n-1}\Pi^\mathsf{T}=\text{blkdiag}(1_{v_{1} \times v_1},1_{v_2 \times v_2},\cdots,1_{v_{r} \times v_r})$ exists, where $v_i$ is the number of nodes belonging to the $i$-th connected  component of $\mathcal{G}(R)$ for each $i \in \mathbb{N}_{[1,r]}$ and $r$ is the number of connected components of $\mathcal{G}(R)$.}
 \end{proof}

Now, we are ready to  present the proof of Theorem~\ref{th:feasible_restriction}.
\begin{proof}
 \emph{$1) \Rightarrow 2)$ Since $\mathcal{P}_{T,R^{n-1}}$ is feasible, there exist $X \succ 0$ in $\text{Sparse}(R^{n-1})$ and $Y \in \text{Sparse}(T)$ such that \eqref{eq:Lyap_X} holds. Upon defining  $P=X^{-1} \succ 0$ and $K = YP$, \eqref{eq:Lyap_X} can be written into (\ref{eq:Lyap_P}) by pre-and post-multiplying by $P$. From $(\ref{eq:Closedloop})$ and $(\ref{eq:Lyap_P})$ we derive that $V(x(t))=x(t)^\mathsf{T}Px(t)$ is a Lyapunov function for the closed-loop system. 
  The rest is to reveal this Lyapunov function is separable.  Since $X \in \text{Sparse}(R^{n-1})$, we have that $P \in \text{Sparse}((R^{n-1})^{n-1})=\text{Sparse}(R^{n-1})$ from the first statement of Lemma~\ref{le:grows_fullest} in the Appendix. 
  Since $R$ is symmetric, $R\geq I_n$ and $P \in \text{Sparse}(R^{n-1})$, it follows from Lemma~\ref{le:symmetry_necessity} that there exists a permutation matrix $\Pi$ such that $\Pi P\Pi^\mathsf{T}$ satisfies~\eqref{eq:separableLy}, indicating that the Lyapunov function $V(x)$ is separable into $r$ addends, where $r$ is the number of connected components of $\mathcal{G}(R)$.}

    \emph{$2) \Rightarrow 1)$  
Define $X=P^{-1}$. Since $P \in \text{Sparse}(R^{n-1})$, we have that $X \in \text{Sparse}((R^{n-1})^{n-1})=\text{Sparse}(R^{n-1})$ from the first statement of Lemma~\ref{le:grows_fullest}. 
 Pre-and post-multiplying (\ref{eq:Lyap_P})  by $X$ we obtain that (\ref{eq:Lyap_X}) is solved with $X=P^{-1}$ and $Y$ as per (\ref{eq:Lyap_P}). Since a matrix $Z$ such that $Z-YX^{-1}Y^\mathsf{T} \succeq 0$ exists for any fixed $Y$, $X$, then problem $\mathcal{P}_{T,R^{n-1}}$ is feasible.}


\end{proof}


 We remark that Theorem~\ref{th:sparsity_invariance} and Theorem~\ref{th:feasible_restriction} offer new insight into the core challenges of distributed control. First, we established the theoretical boundaries of all approaches based on the general idea of sparsity invariance (\ref{eq:sparsity_invariance}), by showing that every convex restriction $\mathcal{P}_{T,R}$ of problem $\mathcal{P}_K$ based on (\ref{eq:sparsity_invariance}) is necessarily subject to $T\leq S$ and $TR^{n-1}\leq S$. Second, we built a direct control-theoretical interpretation of feasibility for these convex restrictions through existence of a separable quadratic Lyapunov function as per Theorem~\ref{th:feasible_restriction}, whereas approaches based on nonlinear and polynomial optimization \cite{lin2013design,wang2018convex} might be difficult to interpret.
 
 
The importance of these new insights is illustrated by a simple example, which shows that requiring $Y \in \text{Sparse}(S)$ and the Lyapunov matrix $P=X^{-1}$ to be diagonal as proposed in \cite{geromel1994decentralized, rubio2013static} fails to compute a feasible controller for  general systems. Instead, appropriately loosening the separability requirement on the Lyapunov function to $r<n$ addends can restore feasibility with good performance, as predicted by Theorem~\ref{th:sparsity_invariance} and Theorem~\ref{th:feasible_restriction}.  

\begin{example}[Restoring Feasibility]
\label{ex:num}
\emph{Consider an unstable three dimensional  continuous-time linear system (\ref{eq:DynamicsGlobal}) with
\begin{equation*}
A=\begin{bmatrix}
2&1&5\\0&-1 &1\\-1&1&0.5
\end{bmatrix}, ~ B=\begin{bmatrix}
1&-1&0\\0&0&-1\\0&0&1
\end{bmatrix}, ~ H=I_3\,.
\end{equation*}
The performance signal (\ref{eq:performance_signal}) is defined  with $C=\begin{bmatrix}
I_{3}&0_{3 \times 3}
\end{bmatrix}^\mathsf{T}$ and $D=\begin{bmatrix}
0_{3 \times 3}&I_3
\end{bmatrix}^\mathsf{T}$. We consider problem $\mathcal{P}_K$, where $\mathcal{S}=\text{Sparse}(S)$ is chosen with
\begin{equation*}
 S=\begin{bmatrix}
1&1&0\\1&1&1\\0&1&1
\end{bmatrix}\,.
\end{equation*}
Note that this system is not divided into interconnected subsystems, unless we interpret scalar states as subsystems. In this case, the works \cite{conte2012distributed, han2017hierarchical, zheng2017convex} suggest forcing the Lyapunov matrix to be diagonal as per \cite{geromel1994decentralized, rubio2013static}. Following these previous approaches, we first consider the convex program $\mathcal{P}_{S,I_3}$ (where we choose $T=S$ and $R$ to be diagonal), which is a convex restriction of $\mathcal{P}_K$ according to Theorem~\ref{th:sparsity_invariance}. We cast and solve this  convex program using SeDuMi \cite{sturm1999using} and YALMIP~\cite{YALMIP}. However, we verify that no feasible solution is found for the considered instance.}

 \emph{We then use our sparsity invariance approach as follows: let $T<S$ and $R$ be chosen as
\begin{equation*}
T=\begin{bmatrix}
1&1&0\\1&1&1\\0&0&1
\end{bmatrix}\,, \quad R=\begin{bmatrix}
1&1&0\\1&1&0\\0&0&1
\end{bmatrix}\,.
\end{equation*}
It can be easily verified that $TR^{n-1} <S$. Hence, $T$ and $R$ satisfy sparsity invariance as per Theorem~$\ref{th:sparsity_invariance}$ and $\mathcal{P}_{T,R^{n-1}}$ is a convex restriction of $\mathcal{P}_K$. Solving $\mathcal{P}_{T,R^{n-1}}$ with SeDuMi \cite{sturm1999using} and YALMIP~\cite{YALMIP} yields the following structured stabilizing controller $K$ with the corresponding Lyapunov matrix $P$ (rounded to the second decimal digit)
\begin{equation*}
	\medmuskip=0mu
\thinmuskip=0mu
\thickmuskip=0mu
K=\begin{bmatrix}
-4.29&3.38&0\\-0.82&1.73&-0.47\\0&0&-8.30
\end{bmatrix}\,,~~P=\begin{bmatrix}
0.50&-0.40&0\\-0.40&0.89&0\\0&0&6.50
\end{bmatrix}\,.
\end{equation*}
 The achieved $\mathcal{H}_2$ norm for the closed-loop system is $5.74$. For comparison, the optimal centralized solution yields an $\mathcal{H}_2$ norm of $3.38$.
 Our generalized convex restriction approach reveals that the closed-loop system admits a Lyapunov function which is separable in two components, whereas a fully separable Lyapunov function cannot be found.}
\end{example}

\section{Optimized Lyapunov Sparsities}
\label{se:optimized}

Theorem~\ref{th:sparsity_invariance} identifies all the convex restrictions of $\mathcal{P}_K$ that are based on the sparsity invariance idea (\ref{eq:sparsity_invariance}).
Among these, one may be interested in finding the convex restriction $\mathcal{P}_{T,R^{n-1}}$  which yields the best performing feasible solution. To this end, it is clear from Theorem~\ref{th:sparsity_invariance} that one could simply solve $\mathcal{P}_{T,R^{n-1}}$ for each $T\leq S$ and $R$ such that $TR^{n-1}\leq S$, then select the best result. However, this trivial approach may not be tractable in general, as one would need to solve a convex program for each admissible $T$ and $R$. Even if a certain $T\leq S$ is fixed for simplicity, one would need to solve a number of convex programs that is exponential in $n^2$ (one for each admissible $R$ such that $TR^{n-1} \leq S$).

To mitigate the challenge above, we suggest a computationally efficient algorithm that directly computes an optimized choice for $R$ given a fixed $T\leq S$. Our suggested approach is based on designing the symmetric binary matrix $R^\star_T$ that yields the best performing convex restriction $\mathcal{P}_{T,R^\star_T}$ of $\mathcal{P}_K$ among all the symmetric binary matrices $R$ satisfying:
\begin{equation}
\label{eq:restriction}
T{R}^{n-1}\leq T\,.
\end{equation}
Such an $R^\star_T$ can be computed with the following algorithm that has polynomial complexity $O(mn^2)$. 

\emph{Step 1:}  Fix $T\leq S$. For every $j,k \in \mathbb{N}_{[1,n]}$, set $(R_T)_{jk}$ to $1$ or $0$ as follows.
\begin{equation}
\label{eq:procedure1}
(R_T)_{jk}=\begin{cases}
0& \text{if }\, \exists i \in \mathbb{N}_{[1,m]}\; \text{s.t.}\;T_{ik}=0,~T_{ij}=1,\\
1& \text{otherwise.}
\end{cases}\\
\end{equation}
In general, this $R_T$ might be non-symmetric. We restore symmetry with an additional step.

\emph{Step 2:} For every $j,k \in \mathbb{N}_{[1,n]}$, set $(R^\star_T)_{jk}$ to $1$ or $0$ as follows.
\begin{equation}
\label{eq:procedure2}
(R^\star_T)_{jk} = \begin{cases} 1& \text{ if } (R_T)_{jk}=(R_{T})_{kj}=1,\\
0& \text{otherwise.}
\end{cases}
\end{equation}

In Proposition~\ref{pr:procedure} below, whose proof is reported in the Appendix, we show that $R^\star_T$ computed according to (\ref{eq:procedure1}), (\ref{eq:procedure2}) 
 yields an optimized convex restriction for any fixed {$T \leq S$}. We provide additional insight on favourable sparsities for the Lyapunov matrix, by proving that performance is maximized because our algorithm (\ref{eq:procedure1}), (\ref{eq:procedure2}) minimizes the  degree of separability forced on the Lyapunov function.
\begin{proposition}
\label{pr:procedure}
Given a binary matrix $T\leq S$, let us restrict our attention to the set of all symmetric binary matrices $R=R^{n-1} \geq I_n$ such that $TR^{n-1}\leq T$. Consider the following statements.
\begin{enumerate}
\item Matrix $R^\star_T$ is computed according to the two-step procedure (\ref{eq:procedure1}), (\ref{eq:procedure2}).
\item The graph $\mathcal{G}(R^\star_T)$ has the minimal number of connected components, thus minimizing the degree of separability forced on a Lyapunov function for the closed-loop system.
\item $\mathcal{P}_{T,R^\star_T}$ is the best performing convex restriction of $\mathcal{P}_K$ among the problems $\mathcal{P}_{T,R^{n-1}}$.
\end{enumerate}
Then $1) \Leftrightarrow 2) \Rightarrow 3)$.
\end{proposition}



We remark that, despite its simplicity, the algorithm (\ref{eq:procedure1}), (\ref{eq:procedure2}) guarantees that a convex restriction at least as performing as that of \cite{geromel1994decentralized, rubio2013static} is obtained by simply choosing $T=S$, due to the fact that $R^\star_T\geq I_n$ by construction. However, we show through the examples of Section~\ref{se:numerics} that it is possible to exploit insight into the specific dynamical system under investigation to obtain better performing choices for $T< S$.

\section{Network Example}
\label{se:numerics}
In this section, we present an illustrative example to validate our results on improving the performance with respect to previous approaches. All instances of problem $\mathcal{P}_{T,R^{n-1}}$  were solved using SeDuMi \cite{sturm1999using} and YALMIP~\cite{YALMIP}, on a computer equipped with a 16GB RAM and a 4.2 GHz quad-core Intel i7 processor.

 Motivated by \cite{zheng2017convex,lin2013design}, we consider an $n \times n$ mesh network of unstable nodes. We assume that each node is a second-order system coupled with  its neighbours in the mesh through a factor $\alpha >0$ as follows:
\begin{equation*}
\dot{x}_i=\begin{bmatrix}1&1\\1&2\end{bmatrix}x_i+\sum_{j \in \mathcal{N}_{\text{mesh}}(i), ~j \neq i}\alpha x_j+\begin{bmatrix}
0\\1
\end{bmatrix}(w_i+u_i)\,,
\end{equation*}
where $\mathcal{N}_{\text{mesh}}(i) \subseteq \mathbb{N}_{[1,n^2]}$ is the set of neighbours of node $i$ according to the mesh topology on the left of Figure~\ref{fig:mesh}. The global system dynamics can be written as (\ref{eq:DynamicsGlobal}), where matrix $A$ is divided into blocks $[A]_{ij}$ of dimension $2 \times 2$ such that
\begin{align*}
&[A]_{ii}=\begin{bmatrix}1&1\\1&2\end{bmatrix},\quad [A]_{ij}=\alpha I_2,\\
&\forall i\in \mathbb{N}_{[1,n^2]}, \quad \forall j \in \mathcal{N}_{\text{mesh}}(i),~j \neq i\,.
\end{align*}
Clearly, we have $B=H=I_{n^2} \otimes \begin{bmatrix}
0&1\end{bmatrix}^\mathsf{T}$, where $\otimes$ denotes the Kronecker product.

We consider a scenario where some nodes in a centralized network remain isolated from the rest of the network. 
The isolated nodes can only use the information of their neighbours in the plant graph.  We let $L \in \mathbb{N}_{[0,n^2]}$ be the number of nodes having full information about the states of all the other nodes, while the remaining $n^2-L$ isolated nodes can only measure the states of their nearest neighbour in the mesh topology (see the left side of Figure~\ref{fig:mesh} for illustration).  For example, when $L=0$ we recover the example of \cite{zheng2017convex,lin2013design}, and $L=n^2$ corresponds to a standard centralized control problem. The resulting information structure is encoded in $S \in \{0,1\}^{n^2 \times 2n^2}$.
\begin{figure}[t]
	      \centering \setlength{\belowcaptionskip}{-15pt}
	      	\includegraphics[width=0.4\textwidth]{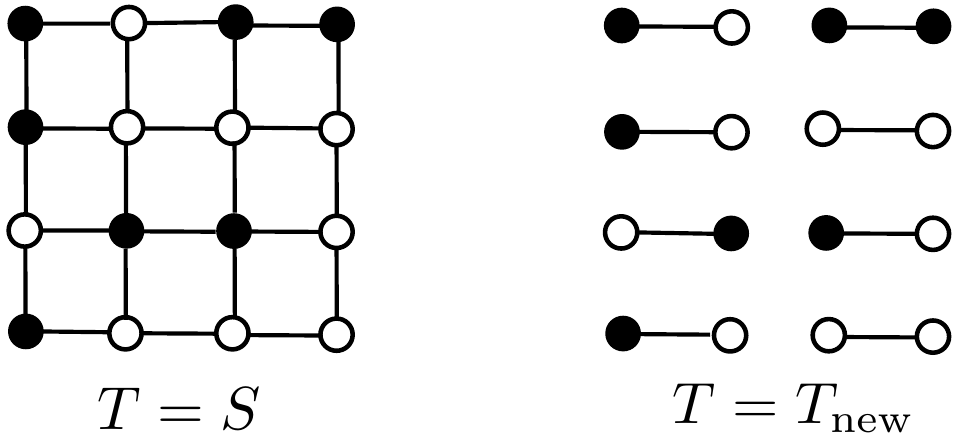}
	      		\caption{ A networked system with $16$ subsystems that are dynamically coupled according to the mesh structure on the left. Black circles represent subsystems having full state information, while  empty circles represent subsystems which  can only measure the states of their neighbours.  A case with $L=7$ is represented. The figure on the left depicts the choice  $T=S$, while the figure on the right represents the choice $T=T_{\text{new}}$ which is based on the maximal cliques contained in the mesh graph. }
	      			      		\label{fig:mesh}
	\end{figure}
%
%
%
 The control objective is to minimize the $\mathcal{H}_2$ norm of the transfer function from disturbances to the performance signal 
 (\ref{eq:performance_signal})   defined by the matrices $C=\begin{bmatrix}
I_{2n^2}&0_{n^2 \times 2n^2}^\mathsf{T}
\end{bmatrix}^{\mathsf{T}}$, $D=\begin{bmatrix}
0_{2n^2 \times n^2}^\mathsf{T}&I_{n^2}
\end{bmatrix}^{\mathsf{T}}$.

For the simulation, we considered a grid of $4 \times 4$ nodes. We first solved the convex restrictions of $\mathcal{P}_K$ obtained by using the approach of requiring the Lyapunov matrix to be block-diagonal \cite{conte2012distributed,zheng2017convex, han2017hierarchical}, where each of the $16$ blocks must have dimension $2 \times 2$, and letting  $Y \in \text{Sparse}(S)$. For every $L \in \mathbb{N}_{[0,16]}$, the result is shown as a red line in Figure~\ref{fig:trajectories}. It can be noticed that, despite relaxing the structural constraints on $Y$ as $L$ increases, the bottleneck in performance remains the (block) diagonal assumption on the Lyapunov matrix.

We then  used the sparsity invariance approach with the optimized Lyapunov sparsity  $R^\star_T$ computed as per (\ref{eq:procedure1}), (\ref{eq:procedure2}). First, we fixed $T=S$ and solved $\mathcal{P}_{S,R^\star_S}$.   We report the results for each $L$ as a green line in Figure~\ref{fig:trajectories}.
The performance improvement is consistent 
with Proposition~\ref{pr:procedure}, where we linked minimized separability of the Lyapunov function with optimized performance.  

By exploiting the structure of the networked dynamical system, it is possible to choose $T<S$ to yield higher performance than the simple choice $T=S$. To validate this observation, we chose $T_{\text{new}}<S$ as follows
\begin{equation*}
T_{\text{new}}=T_{\text{cliques}}+T_{\text{add}}=I_{8 \times 8} \otimes 1_{2 \times 4} +  T_{\text{add}}\,,
\end{equation*}
where $T_{\text{add}}$ indicates the additional information available to the $L$ centralized agents.  The matrix $T_{\text{cliques}}$ was chosen to encode the maximal cliques within the mesh topology as shown on the right side of Figure \ref{fig:mesh}. The intuition into choosing $T_{\text{cliques}}$ is that it shows an efficient trade-off between reducing separability requirements on the Lyapunov function for the closed-loop system and restricting the structure of $T$.  Indeed, the blocks of $R^\star_{T_{\text{cliques}}}$ corresponding to isolated cliques can be dense while still ensuring $T_{\text{cliques}}(R^\star_{T_{\text{cliques}}})^{n-1}\leq T_{\text{cliques}}$. By using the optimized Lyapunov matrix $R^\star_{T_{\text{new}}}$, we solved problem $\mathcal{P}_{T_{\text{new}},R^\star_{T_{\text{new}}}}$. We report the result for every $L$ as a blue line in Figure~\ref{fig:trajectories}. The simulation shows that the choice  $T=T_{\text{new}}<S$ leads to improved performance for $L>3$. For $L\leq 3$, the performance improvement obtained by loosening the separability requirements on the Lyapunov function is not sufficient to compensate for the more restrictive structural constraints on $Y \in \text{Sparse}(T_{\text{new}})$.
\begin{figure}[t]
    \centering \setlength{\belowcaptionskip}{-9pt}
    \includegraphics[width=0.55\linewidth]{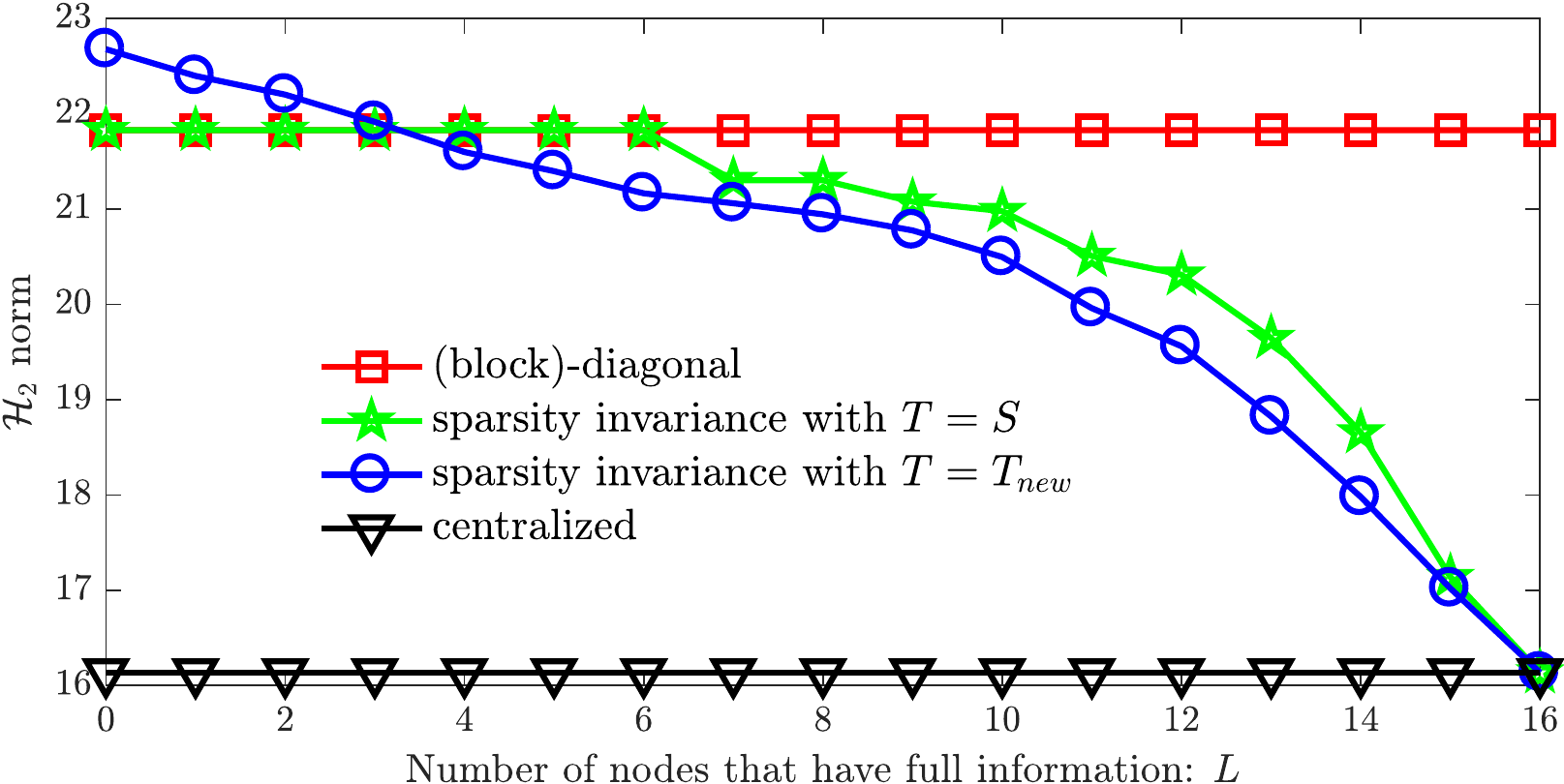}	      		
    \flushleft
    \caption{Performance comparison for the proposed sparsity invariance approach, the block-diagonal strategy, and the centralized control. The agents with full information were chosen randomly according to the following succession, where we number the agents in the mesh from left to right  starting from the top row: $14$, $15$, $3$, $11$, $2$, $5$, $9$, $16$, $8$, $13$, $7$, $1$, $12$, $6$, $10$, $4$.}
    \label{fig:trajectories}
    \vspace{-3mm}
\end{figure}


\section{Conclusions}
With the aim of improving feasibility and performance of approaches based on computing a block-diagonal Lyapunov function for the closed-loop system \cite{geromel1994decentralized, zheng2017convex, rubio2013static, conte2012distributed, han2017hierarchical}, we characterized a generalized class of feasible convex restrictions of the static optimal distributed control problem   based on the concept of separable Lyapunov functions. We validated our main results through numerical examples.

Several directions remain open for future research. First, with a similar spirit to the results of \cite{rantzer2015scalable} about existence of diagonal Lyapunov functions for positive systems, our findings motivate identifying  more general classes of dynamical systems for which stability is equivalent to existence of a separable Lyapunov function.
Second, the efficient algorithm we suggested to design performing structures for the Lyapunov matrix relies on the approximation that the structure of one of the two matrix factors is fixed beforehand.  Hence, it would be interesting to develop efficient heuristics that identify high performing sparsities for both matrix factors simultaneously.
Last, since all separable Lyapunov matrices can be permuted to have a block-diagonal structure, it is relevant to explore the connections with \cite{sootla2017block, zheng2018scalable} for scalable synthesis of distributed controllers.
\balance
\section*{Appendix}
\subsection{Proof of Theorem~\ref{th:sparsity_invariance}}

\setcounter{lemma}{0}
\renewcommand{\thelemma}{\Alph{subsection}\arabic{lemma}}
\label{app:sparsity_invariance}
The proof relies on the following Lemmas.
\begin{lemma}
\label{le:grows_fullest}
Let $R \in \{0,1\}^{n \times n}$  with $R \geq I_n$. Then, 
\begin{enumerate}
\item for any invertible $X$ in $\text{Sparse}\left(R\right)$, we have
\begin{equation*}
\text{Struct}\left(X^{-1}\right) \leq R^{n-1}\,.
\end{equation*}
\item there exists an invertible matrix $X\in \text{Sparse}(R)$ such that
\begin{equation*}
\text{Struct}\left(X^{-1}\right)=R^{n-1}\,.
\end{equation*}
\end{enumerate}
\end{lemma}
\begin{proof}
\emph{Suppose $X \in \text{Sparse}(R)$ is invertible. By Cayley-Hamilton's theorem, $\sum_{i=0}^{n}\lambda_iX^i =0$ where $\{\lambda_i\}_{i=0}^{n}$ are the coefficients of the characteristic polynomial of $X$ and $\lambda_0=\det{X}\neq 0$. By pre-multiplying by $X^{-1}$ and rearranging the terms we obtain
\begin{equation}
\label{eq:inverse_sum}
X^{-1}=-\lambda_0^{-1}(\lambda_1 I+\lambda_2X+\lambda_3X^2+\cdots+\lambda_nX^{n-1})\,.
\end{equation}
Since $R\geq I$ we have that $R^a \geq R^b$ for every integer $a\geq b$. Hence,  $\lambda_i X^i \in \text{Sparse}\left(R^{n-1}\right)$ for every $i$ and the first statement follows by (\ref{eq:inverse_sum}).}

\emph{For the second statement, we iteratively construct $X$ starting from $X=I_n$. Let $\alpha \in \mathbb{R}$. Define $\tilde{X}=X+\alpha e_i e_j^\mathsf{T}$.  Let  $X^{-1}_{:,i} \in \mathbb{R}^{n \times 1}$  and  $X^{-1}_{j,:}\in \mathbb{R}^{1 \times n}$ be the $i$-th column  and the $j$-th row of $X^{-1}$ respectively, and let $X^{-1}_{ij}$ be the entry $(i,j)$ of $X^{-1}$. Using the Sherman-Morrison identity \cite{sherman1950adjustment}, we obtain
\begin{equation}
\label{eq:adding_to_i}
\tilde{X}^{-1}_{i,:}=X^{-1}_{i,:}-\frac{\alpha X^{-1}_{ii}}{1+\alpha X^{-1}_{ji}} X^{-1}_{j,:}\,.
\end{equation}
 From (\ref{eq:adding_to_i}), it is easy to verify that, for any $i$ and $\alpha \in \mathbb{R}$, if $X^{-1}_{ii}\neq 0$, then $\tilde{X}^{-1}_{ii} \neq 0$. It follows that by choosing $\alpha$ such that
\begin{align}
\label{eq:alpha_condition}
&\alpha X^{-1}_{ji} \neq -1 \text{ and }\alpha \left(X^{-1}_{ii} X^{-1}_{jk}-X^{-1}_{ji}X^{-1}_{ik}\right) \neq X^{-1} _{ik}\,,\nonumber\\
& \forall k\text{ subject to }X^{-1}_{jk}\text{ and }X^{-1}_{ik} \text{ are not both null}\,,
\end{align}
we obtain that
\begin{equation}
\label{eq:augment_structure}
\text{Struct}\left(\tilde{X}^{-1}_{i,:}\right)=\text{Struct}\left(X^{-1}_{i,:}\right)+\text{Struct}\left( X^{-1}_{j,:}\right)\,.
\end{equation}
 The condition (\ref{eq:alpha_condition}) is derived by setting the right hand side of (\ref{eq:adding_to_i}) to be different from $0$ for every $k$ such that  $X^{-1}_{ik}$ and $X^{-1}_{jk}$ are not both null. When both are null, (\ref{eq:adding_to_i}) reveals $\tilde{X}^{-1}_{ik}=0$, coherent with (\ref{eq:augment_structure}). The structural augmentation (\ref{eq:augment_structure}) is exploited in the algorithm below.
	\begin{algorithm}[H]
\label{alg:buildXinv}
  \begin{algorithmic}[1]
  \State Set $X=I_n$
  \Repeat
  \Comment{max. $(|R|-n)(n-1)$ iterations}
  \For{ each $(i,j)$ such that $i\neq j$ and $R_{ij}=1$ }
  	\State Choose $\alpha$  according to (\ref{eq:alpha_condition})
  	\State $X\leftarrow X+\alpha e_ie_j^\mathsf{T}$
  \EndFor
  \Until{$\text{Struct}(X^{-1})=R^{n-1}$}
  \State Return $X$
  \end{algorithmic}
\end{algorithm}
The algorithm  returns a matrix $X$ such that $\text{Struct}(X^{-1}) =R^{n-1}$. Specifically, by exploiting (\ref{eq:augment_structure}) we obtain that $\text{Struct}(X^{-1}) \geq R^{s}$ at the end of the $s$-th iteration of the ``repeat-until'' cycle.}

%
%

\end{proof}

\begin{lemma}
\label{le:T}
Let $T \in \{0,1\}^{m \times n}$ and $R \in \{0,1\}^{n \times n}$, and  $\text{Struct}(
W)=R$. Then, there exists $Z \in \text{Sparse}(T)$ such that
$$\text{Struct}(ZW)=TR\,.$$
\end{lemma}
\begin{proof}
\emph{Let $Z$ be any matrix in $\text{Sparse}(T)$. Assume that $\text{Struct}(ZW)<TR$. Then, for some $(i,j,k)$ we have that $ZW_{ij}=0$ and $T_{ik}=R_{kj}=1$. We know by hypothesis that  $W_{kj}\neq 0$. Since $\sum_{l=1}^nZ_{il}W_{lj}=0$, it is sufficient to update $Z_{ik}$ with $Z_{ik}+\alpha$ for any $\alpha \neq 0$ to guarantee that $ZW_{ij} \neq 0$. Furthermore, by choosing $\alpha \neq -\frac{ZW_{is}}{W_{ks}}$ for all $s$ such that $ZW_{is}\neq 0$, we avoid that adding $\alpha$ to $Z_{ik}$ brings $ZW_{is}$ to $0$ when $ZW_{is}\neq 0$. Hence, it is always possible to choose $k$ and $\alpha$ such that $ZW+\alpha e_ie_k^\mathsf{T}>ZW$ and $Z \in \text{Sparse}(T)$. By iterating the procedure for all $(i,j)$ such that $\text{Struct}(ZW)_{ij}<TR_{ij}$, we converge to $\text{Struct}(ZW)=TR$.}

\end{proof}

We are now ready to prove Theorem~\ref{th:sparsity_invariance}.

$1) \Rightarrow 2)$: We prove by contrapositive. First, suppose that $TR^{n-1} \not \leq S$. By the second statement of  Lemma~\ref{le:grows_fullest} it is possible to select $X\in \text{Sparse}(R)$ such that $\text{Struct}(X^{-1})=R^{n-1}$. By the latter and Lemma~\ref{le:T}, we can select $Y \in \text{Sparse}(T)$ such that $\text{Struct}\left(YX^{-1}\right)=TR^{n-1}$, or equivalently $YX^{-1} \not \in \mathcal{S}$. Next, suppose that $T\not \leq S$. Since $R\geq I_n$ by hypothesis, then $TR \not \leq S$ and $TR^{n-1} \not\leq S$. Hence, the same reasoning applies.

$2) \Rightarrow 1)$: Let $X \in \text{Sparse}(R)$ be invertible.  By Lemma~\ref{le:grows_fullest} we know that $X^{-1} \in \text{Sparse}(R^{n-1})$. Now let $Y \in \text{Sparse}(T)$. Since $TR^{n-1} \leq S$, we have $YX^{-1} \in \mathcal{S}$.

$1) \Rightarrow 3)$: If (\ref{eq:sparsity_invariance}) holds, clearly $\mathcal{P}_{T,R}$ is a restriction of the non-convex problem where $YX^{-1} \in \mathcal{S}$. The latter problem is equivalent to $\mathcal{P}_K$. Hence,  $\mathcal{P}_{T,R}$ is a restriction of $\mathcal{P}_K$. Since $TR^{n-1}\leq S$, by Cayley-Hamilton and $R\geq I_n$ we have that $(R^{n-1})^{n-1}$ is also such that $T(R^{n-1})^{n-1} \leq S$. Hence, $\mathcal{P}_{T,R^{n-1}}$ is a convex restriction of $\mathcal{P}_K$.

\subsection{Proof of Proposition~\ref{pr:procedure}}
\label{app:procedure}

$1) \Rightarrow 2)$: It is easy to verify $T(R^\star_T)^{n-1} \leq T$. Indeed, $T(R_T)^{n-1} \leq \cdots \leq  T(R_T)\leq T$ by  (\ref{eq:procedure1}) and $R^\star_T \leq R_T$ by (\ref{eq:procedure2}).  Also, $R^\star_T \geq I_n$ by construction.
 Now, consider any binary symmetric $R$ such that $R\geq I_n$ and $TR^{n-1} \leq T$. Since $R^{n-1}$ is symmetric, we have that whenever $T_{ik}=0$ and $T_{ij}=1$, then $(R^{n-1})_{jk}=(R^{n-1})_{kj}=0$. This implies that $R^{n-1} \leq R^\star_T$ by definition (\ref{eq:procedure1}), (\ref{eq:procedure2}).  It follows that $(R^\star_T)^{n-1} \leq R^\star_T$. Since $(R^\star_T)^{n-1}\geq R^\star_T$ because $R^\star_T \geq I_n$, we conclude that $(R^\star_T)^{n-1} = R^\star_T$. Now, consider the following optimization problem with binary variables:
 {
	\medmuskip=0mu
\thinmuskip=0mu
\thickmuskip=0mu
	\begin{alignat*}{3}
	&\minimize_{R \in \{0,1\}^{n \times n}} &&   ~~~r\nonumber\\
	&~\st &&~~R=R^\mathsf{T},~~ R\geq I_n,~~R=R^{n-1}, ~~TR^{n-1} \leq T\,,\\
	&~&&~~\mathcal{G}(R)\text{ has }r\text{ connected components}\,,
	\end{alignat*}}
where we aim at minimizing the number of connected components of the graphs $\mathcal{G}(R)$ under the assumptions on $R$ stated in the proposition statement. We have shown above that $R^\star_T$ is feasible for this problem. We have also shown that any other feasible $R$ is such that $R=R^{n-1}\leq R^\star_T$. It follows that $\mathcal{G}(R)$ is a subgraph of $\mathcal{G}(R^\star_T)$ for every feasible $R$. Since $\mathcal{G}(R)=\mathcal{G}(R^{n-1})$ and $\mathcal{G}(R^\star_T)$ are graphs consisting of separate complete subgraphs, it follows that either $\mathcal{G}(R)$ is equal to $\mathcal{G}(R^\star_T)$ or $\mathcal{G}(R)$ has strictly more connected components than $\mathcal{G}(R^\star_T)$. Statement $2)$ follows.

$2) \Rightarrow 1)$ We prove by contrapositive. Let $R^\star_T$ be computed as per (\ref{eq:procedure1}), (\ref{eq:procedure2}). We have proven above that $(R^\star_T)^{n-1}=R^\star_T$. Now consider the optimization problem introduced in the proof that $1) \Rightarrow 2)$. If $R$ is not feasible, then it does not satisfy the assumptions on $R$ made in the proposition statement. Hence, take any feasible $R$ for the optimization problem such that $R \neq R^\star_T$. We have three cases:
\begin{enumerate}[(i)]
\item $R<R^\star_T$,
\item $R \not \leq R^\star_T$ and $R \not \leq R_T$,
\item $R \not \leq R^\star_T$ and $R \leq R_T$.
\end{enumerate}
In case i), $\mathcal{G}(R)$ is a strict subgraph of $\mathcal{G}(R^\star_T)$. Both $\mathcal{G}(R)$ and $\mathcal{G}(R^\star_T)$ consist of complete subgraphs. This implies that $r>r^\star$, where $r$ and $r^\star$ are the number of  connected components of $\mathcal{G}(R)$ and $\mathcal{G}(R^\star_T)$ respectively.
 Hence, $R$ cannot be an optimal solution of the optimization problem. In case ii) we have that $TR^{n-1} \not \leq T$ by construction of $R_T$. Hence, $R$ cannot be feasible for the optimization problem. In case (iii), notice that there is no symmetric $R$ such that $R \not \leq R^\star_T$ and $R\leq R_T$ by construction. This contradicts the hypothesis that $R \leq R_T$.   We conclude that if $R \geq I_n$, $R=R^{n-1}$ and $TR^{n-1} \leq T$ is not chosen according to procedure (\ref{eq:procedure1}), (\ref{eq:procedure2}), then the number of connected components of $\mathcal{G}(R)$ is not minimized. 


 $1) \Rightarrow 3)$  Let $R$ be any symmetric binary matrix such that $R=R^{n-1}\geq I_n$, $TR^{n-1} \leq T$ and $R \neq R^\star_T$, where $R^\star_T$ is computed according to (\ref{eq:procedure1}), (\ref{eq:procedure2}). We have shown that $R^{n-1}< R^\star_T$. Hence, $\text{Sparse}(R^{n-1}) \subset \text{Sparse}(R^\star_T)$. Since we require $X \in \text{Sparse}(R^{n-1})$,  we conclude that it is never convenient in terms of performance and feasibility to solve $\mathcal{P}_{T,R^{n-1}}$ instead of $\mathcal{P}_{T,R^\star_T}$ as a convex restriction of $\mathcal{P}_K$.

	\bibliographystyle{IEEEtran}
	\bibliography{IEEEabrv,references}
	
\end{document}